\documentclass[final,5p]{elsarticle}

\usepackage{amsmath,amssymb,amsthm}
\usepackage{mathtools}
\usepackage{url}
\usepackage{hyperref}
\usepackage{placeins}
\usepackage{algorithm}
\usepackage{algpseudocode}

\newtheorem{lemma}{Lemma}
\newtheorem{theorem}{Theorem}

\journal{Information Processing Letters}
\DeclareMathOperator{\supp}{supp}

\begin{document}

\begin{frontmatter}

\title{Linear KL-Optimal Frequency Normalisation}

\author{Kamila Szewczyk}
\address{Algorithmic Bioinformatics group, Saarland University,
Saarbr\"ucken, Germany}
\ead{szewczyk@cs.uni-saarland.de}

\begin{keyword}
Entropy coding \sep asymmetric numeral systems \sep
Kullback-Leibler divergence \sep integer normalisation \sep
separable concave optimisation
\end{keyword}

\begin{abstract}
Fast implementations of range coding and asymmetric numeral systems (ANS) owe their excellent performance to replacing slow division instructions by bit-shifts in their encoding and decoding algorithms. This is possible when the frequency distribution of symbols is normalised such that it sums to a power of two. However, such normalisation typically introduces a marginal increase in the Kullback-Leibler divergence between the original and the normalised distribution, leading to a worse compression ratio. We show that the currently used methods for frequency normalisation are suboptimal in both their running time and the achieved Kullback-Leibler divergence. We propose a new method for frequency normalisation that is asymptotically linear in the number of symbols and achieves the smallest possible Kullback-Leibler divergence between the original and the normalised distribution. The method is based on a solution to a separable concave optimisation problem, which may be of independent interest.
\end{abstract}

\end{frontmatter}

\section{Introduction}

Let $\mathcal X$ be a finite alphabet with symbol counts $c_x \in \mathbb{Z}_{\ge 0}$ for $x \in \mathcal X$. The total count is $T = \sum_{x \in \mathcal X} c_x$. The empirical law is thus given by $p_x = c_x / T$. Fix some $M \ge 1$, in practice a large power of two~\cite{Duda2013}, and let $q_x$ be a normalised frequency distribution such that $\sum_{x \in \mathcal X} q_x = 1$ and $q_x = m_x / M$ for some integer $m_x \in \mathbb{Z}_{\ge 0}$, with $m_x > 0$ if $c_x > 0$ (i.e., the support of $q$ is the same as that of $p$). Coding the source as described by the distribution $p$ with the normalised distribution $q$ results in a code that has at least the length of $\mathrm H(p) + \mathrm D(p \| q)$ nats per symbol, where $\mathrm H(p)$ is the entropy of $p$ and $\mathrm D(p \| q)$ is the Kullback-Leibler divergence between $p$ and $q$~\cite{KullbackLeibler1951,CoverThomas2006}:

\begin{equation}
\begin{aligned}
\mathrm D(p \| q) &= \sum_{x \in \mathcal X} p_x \log \frac{p_x}{q_x} = \sum_{x \in \mathcal X} (c_x / T) \log \frac{c_x / T}{m_x / M}\\
&= \sum_{x \in \mathcal X} (c_x / T) \log \frac{c_x M}{m_x T}.
\end{aligned}
\end{equation}

We wish to find such $m_x$ that minimises $\mathrm D(p \| q)$. This problem is equivalent to maximising $\Phi(m) = \sum_{x \in \supp \mathcal X} c_x \log m_x$, a separable concave integer allocation problem~\cite{FedergruenGroenevelt1986,Hochbaum1994}.

For a specific allocation, the unit-move gains (i.e., the gains from moving one quantum, $1/M$, from one symbol to another in the quantised distribution),

\begin{equation}
   \Delta^+_x(j)=c_x\log\tfrac{j+1}{j},\qquad
   \Delta^-_x(j)=c_x\log\tfrac{j}{j-1}
\end{equation}

are strictly decreasing in $j$, and a support-feasible $m$ is optimal iff

\begin{equation}\label{eq:mx}
   \min_{x\in \supp \mathcal{X}:\,m_x\ge2}\Delta^-_x(m_x)\;\ge\;
   \max_{y\in \supp \mathcal{X}}\Delta^+_y(m_y),
\end{equation}

thus a violation of optimality is a strictly improving unit transfer, and the converse follows from the standard exchange argument for concave allocation.

\section{Prior Art}

Each of the algorithms named below forms $s_x = Mp_x = c_x M / T$ for the ideal scaled frequencies, rounds it and attempts to restore $\sum_x m_x = M$ ad-hoc. While fast, none enforces the optimality condition~\eqref{eq:mx}.

Giesen~\cite{ryg_rans} rescales the empirical cumulative counts to total $M$ and
differences them, taking $m_x$ to be the increment of the rescaled cumulative
distribution at $x$, and fixes the support by giving one point of mass to the
items that now lie outside of it by transferring one quantum from the most likely
symbol. The first pass is done in $\mathcal{O}(|\mathcal X|)$, while the repair rescans for the largest symbol once per transferred quantum and thus costs $\mathcal{O}(|\mathcal X| b)$. assuming $b$ symbols are boosted to the
support floor. Bloom~\cite{Bloom2014ANS} is markedly better: the algorithm rounds each to the nearer of $\lfloor s_a\rfloor,\lfloor s_a\rfloor{+}1$ by a geometric-mean test and restores the total with a unidirectional heap
correction, in $\mathcal{O}\!\left(|\mathcal X|+\lvert\sum_a m^{(0)}_a-M\rvert\log
|\mathcal X|\right)$ time for the rounded vector $m^{(0)}$ (worst case
$\mathcal{O}(M\log |\mathcal X|)$). FSE~\cite{ColletFSE} performs the rounding with a small threshold table and assigns the unused
quanta to the most likely symbol, in $\mathcal{O}(|\mathcal X|)$ time (one Barrett
reciprocal~\cite{Barrett1986} and a few integer operations per symbol, the cheapest
normaliser in practice); the so-called ``perfect''
variant~\cite{Collet2014Perfect} runs a decrement heap from the ceiling $\lceil s_a\rceil$ in $\Theta\!\left((\sum_a\lceil s_a\rceil-M)\log |\mathcal X|\right)$
time. We point out that this algorithm is not KL-optimal either, as the ceiling envelope need not contain an optimum. Table~\ref{tab:wit} gives, for each algorithm, a small counterexample on which the optimality condition
\eqref{eq:mx} fails.

\begin{table}[t]
\centering
\caption{For each normaliser, a small input on which it is
KL-suboptimal, with the per-symbol KL gap (nats) to the optimum.}
\label{tab:wit}
\begin{tabular}{llll}
\hline
normaliser & witness & $M$ & gap \\
\hline
Giesen          & $(3,2)$               & $256$ & $6.3\times10^{-6}$ \\
Bloom           & $(3046,2582,4294)$    & $8$   & $2.6\times10^{-5}$ \\
FSE & $(10,3,3)$            & $8$   & $9.5\times10^{-3}$ \\
``Perfect''  & $(22,4,4,4,4,4,4,4,4)$& $16$  & $3.1\times10^{-3}$ \\
\hline
\end{tabular}
\end{table}

\section{An \texorpdfstring{$\mathcal{O}(|\mathcal X|)$}{O(|X|)} exact normaliser}

We first prove the existence of a tight window around the optimal solution, which allows us to restrict the search space to a small interval around the ideal scaled frequencies $s_x$.

\begin{lemma}\label{lem:win}
For every KL-optimal $m^\star$ and every $a\in \supp \mathcal{X}$,
\[
\begin{aligned}
   m^\star_a\;&\ge\;s_a-p_a(|\supp \mathcal X|-2)-1,\\
   m^\star_a\;&\le\;s_a+p_a(|\supp \mathcal X|-2)+1 .
\end{aligned}
\]
\end{lemma}
\begin{proof}
We prove the upper bound. The lower bound is symmetric. Suppose
$m^\star_a>p_a(M+|\supp \mathcal X|-2)+1$. Then
$\sum_{b\in \supp \mathcal{X}\setminus\{a\}}(m^\star_b+1)=M-m^\star_a+|\supp \mathcal X|-1<(1-p_a)(M+|\supp \mathcal X|-2)$,
so, as $\sum_{b\in \supp \mathcal{X}\setminus\{a\}}p_b=1-p_a$, some $b\in \supp \mathcal{X}\setminus\{a\}$
satisfies
$(m^\star_b+1)/p_b<M+|\supp \mathcal X|-2<(m^\star_a-1)/p_a$, i.e.
$c_b/(m^\star_b+1)>c_a/(m^\star_a-1)$. Using $\log(1+x)<x$ and
$\log(1+x)>x/(1+x)$ for $x>0$,
\begin{equation}
\begin{aligned}
   \Delta^-_a(m^\star_a)&=c_a\log\tfrac{m^\star_a}{m^\star_a-1}
   <\tfrac{c_a}{m^\star_a-1}\\
   &<\tfrac{c_b}{m^\star_b+1}
   <c_b\log\tfrac{m^\star_b+1}{m^\star_b}=\Delta^+_b(m^\star_b),
\end{aligned}
\end{equation}
which violates \eqref{eq:mx}. Hence
$m^\star_a\le p_a(M+|\supp \mathcal X|-2)+1=s_a+p_a(|\supp \mathcal X|-2)+1$.
\end{proof}

Let $L_a=\max\{1,\lceil s_a-p_a(|\supp \mathcal X|-2)-1\rceil\}$ and
$U_a=\lfloor s_a+p_a(|\supp \mathcal X|-2)+1\rfloor$. By Lemma~\ref{lem:win} every
optimum lies in $L_a\le m_a\le U_a$, whose total width is
\begin{equation}
\begin{aligned}
   \sum_{a\in \supp \mathcal{X}}(U_a-L_a)&\le\sum_{a\in \supp \mathcal{X}}\bigl(2p_a(|\supp \mathcal X|-2)+2\bigr)\\&=4|\supp \mathcal X|-4 .
\end{aligned}
\end{equation}
Since some optimum has $\sum_{x\in \supp \mathcal{X}}m_x = M$, $0\le D:=\sum_a U_a-M\le\sum_a(U_a-L_a)$.
For any $m$ with $L_a\le m_a\le U_a$ for all $a$, the loss relative to $U$ is
$\Phi(U)-\Phi(m)=\sum_{a}\sum_{j=m_a+1}^{U_a}\Delta^-_a(j)$: lowering $m_a$ from
$U_a$ spends the tickets $\Delta^-_a(j)$ for $j=m_a+1,\ldots,U_a$. Since
$\Delta^-_a(\cdot)$ is decreasing, the cheapest tickets of a row are those with
the largest $j$, so any $D$ globally cheapest tickets form, within each row, a
run ending at $U_a$; they therefore correspond to a vector $m$ obeying the window
and summing to $M$, and it is the one of least loss, hence a global optimum. Algorithm~\ref{alg:win} emits the at most
$4|\supp \mathcal X|-4$ tickets and picks the $D$ cheapest by linear-time selection.

\begin{algorithm}[t]
\caption{Linear KL-optimal normalisation}
\label{alg:win}
\begin{algorithmic}[1]
\Require counts $c$, total $M$
\State $S\gets\{a:c_a>0\}$, $r\gets|S|$, $T\gets\sum_{a\in S}c_a$
\If{$M<r$} \Return ``no finite-KL solution'' \EndIf
\State $\mathcal T\gets\emptyset$, $D\gets-M$
\For{$a\in S$}
   \State $L_a\gets\max\{1,\lceil Mc_a/T-(c_a/T)(r-2)-1\rceil\}$
   \State $U_a\gets\lfloor Mc_a/T+(c_a/T)(r-2)+1\rfloor$
   \State $m_a\gets U_a$;\quad $D\gets D+U_a$
   \For{$j=L_a+1,\ldots,U_a$} add ticket $(\Delta^-_a(j),a)$ to $\mathcal T$ \EndFor
\EndFor
\State select the $D$ cheapest tickets of $\mathcal T$; decrement $m_a$ for each
\State \Return $m$
\end{algorithmic}
\end{algorithm}

\begin{theorem}\label{thm:main}
Algorithm~\ref{alg:win} returns an optimal frequency vector in
$\mathcal{O}(|\mathcal{X}|)$ time under constant time ticket comparisons, and this is asymptotically optimal in terms of time.
\end{theorem}

This algorithm specialises the linear-time selection approach to divisor
methods of apportionment~\cite{ChengEppstein2014,ReitzigWild2023} to
a logarithmic objective. Its running time matches the linear heuristics discussed above, while improving
on Collet-perfect's $\Theta\!\left((\sum_a\lceil s_a\rceil-M)\log |\mathcal X|\right)$
and Bloom's $\mathcal{O}(M\log |\mathcal X|)$, and being the only \textit{exact}
method among them.

Regardless, we note that floating-point evaluation of the tickets $\Delta^-_a$ via $\log(j/(j-1))$ suffices in the count regimes of practical interest, but may become expensive if high precision is required or $j$ is large; in these cases an
exact transcendental-free comparator $j^{c_a}(\ell-1)^{c_b}\gtreqless(j-1)^{c_a}\ell^{c_b}$ is available.

\section{Worst-case cost of normalisation}

The bound of Lemma~\ref{lem:win} offers a good worst-case guarantee on the inherent KL divergence of the optimal normalisation. Let $\rho(p,M)=\min_m D(p\,\|\,q)$ be the optimal quantisation redundancy over feasible $m$. One easy estimate is interior and $\Theta(1/M^2)$: for $p$ in the interior of the simplex, each $s_a$ may be favourably rounded to within $\tfrac12$, and the well-known quadratic expansion $D(p\,\|\,q)=\tfrac12\sum_a(p_a-q_a)^2/p_a+o(\lVert p-q\rVert^2)$, whose leading term is the $\chi^2$ divergence and hence the second-order term of the order-2 Rényi divergence, gives $1/M^2$. The worst case over a fixed support is larger:

\begin{theorem}
\label{thm:redundancy}
For all $r\ge1$ and $M\ge r$,
\begin{equation}
\begin{aligned}
   \sup_{p:\,|\supp \mathcal X|=r}\rho(p,M)
   &=\log\frac{M}{M-r+1} \\
   &=\frac{r-1}{M}+\mathcal{O}\!\left((r/M)^2\right)\ \text{nats},
\end{aligned}
\end{equation}
and the supremum is approached but not attained. The same value results when $m$
is allowed to range over the reals subject to $q_a\ge1/M$, so integrality costs
nothing in the worst case.
\end{theorem}

\begin{proof}
\emph{Upper bound.} For any $p$ the allocation $x_a=1+p_a(M-r)$ is feasible
($x_a\ge1$, $\sum_{a\in \supp \mathcal X}x_a=M$) and
$x_a=(1-p_a)+p_a(M-r+1)\ge p_a(M-r+1)$, so $q_a=x_a/M\ge p_a(M-r+1)/M$ and
$p_a/q_a\le M/(M-r+1)$; hence $D(p\,\|\,q)\le\log\frac{M}{M-r+1}$. Rounding
to $m_a=\lceil p_a(M-r+1)\rceil$ keeps $m_a\ge p_a(M-r+1)$ and
$\sum_a m_a<M+1$, hence $\sum_a m_a\le M$; assigning the remaining units to any
symbol only raises its $q_a$ and so only lowers $D$, and the bound holds for
integer frequencies.
\emph{Lower bound.} Consider the ``almost hot'' distribution
$p_\varepsilon=(1-(r-1)\varepsilon,\varepsilon,\dots,\varepsilon)$, in which one
symbol carries almost all the mass, and any feasible $q$; its ``almost-cold''
mass satisfies $Q=\sum_{a\ge2}q_a\ge(r-1)/M$, since each of the $r-1$ cold
symbols needs $m_a\ge1$. The log-sum
inequality bounds the ``almost-cold'' terms below by
$(r-1)\varepsilon\log\frac{(r-1)\varepsilon}{Q}$, so with $q_1=1-Q$,
\begin{equation}
\begin{aligned}
   D(p_\varepsilon\,\|\,q)\ge
   (1-(r-1)\varepsilon)&\log\frac{1-(r-1)\varepsilon}{1-Q}\\
   +(r-1)\varepsilon&\log\frac{(r-1)\varepsilon}{Q}.
\end{aligned}
\end{equation}
As a function of $Q$ the right side is minimised at $Q=(r-1)\varepsilon$ and
increases thereafter, so for $\varepsilon<1/M$ it is increasing on the feasible
range $Q\ge(r-1)/M$ and is minimised at the endpoint $Q=(r-1)/M$, and there it tends
to $\log\frac{M}{M-r+1}$ as $\varepsilon\to0^+$. That endpoint is realised by the
integer allocation $(M-r+1,1,\dots,1)$, so the limit bounds $\rho$ as well.
\end{proof}

\section{Conclusion}
\label{sec:conclusion}
Integer frequency normalisation for entropy coders is quickly and exactly solvable through methods familiar from separable concave allocation. A two-sided window of width $4|\supp \mathcal X|-4$ around the ideal frequencies contains every optimum, immediately yielding an exact KL-optimal normaliser linear in the support size and provably correct on every input, unlike the existing heuristics; the tools used also bound the unavoidable redundancy to $(|\supp \mathcal X|-1)/M$ nats plus negligible terms in the worst case.

\bibliography{refs}

\end{document}